\documentclass[11pt]{article}

\usepackage{amsmath}
\usepackage{amsthm}
\usepackage{amssymb}
\usepackage{mathrsfs}
\usepackage{graphicx}
\def\div{{\,\rm div \,}}

\def\g{{\,\rm \gamma \,}}

\def\Id{{\,\rm Id \,}}

\def\CC{{\,\rm C\,}}
\def\WW{{\,\rm W\,}}
\def\qiq{{\quad\mbox{in}\quad}}
\def\qaq{{\quad\mbox{at}\quad}}

\def\id{{\,\rm id \,}}

\def\sym{{\,\rm sym \,}}
\def\Ric{{\,\rm Ric\,}}
\def\ii{{\,\rm i \,}}

\def\supp{{\,\rm supp \,}}

\def\B{{\,\cal B \,}}

\def\+M{{\,\rm M^{n\times n}_+ \,}}
\def\tr{{\,\rm tr \,}}

\def\qfq{{\quad\mbox{for}\quad}}

\def\qflq{{\quad\mbox{for all}\quad}}

\def\LL{{\,\rm L \,}}

\def\ii{{\,\rm i \,}}

\def\supp{{\,\rm supp \,}}

\def\C{{\cal C }}

\def\lam{\lambda}

\def\ol{\overline}

\def\E{{\cal E}}

\def\E{{\cal E}}
\def\L{{\cal L}}

\def\A{{\cal A}}

\def\P{{\cal P}}

\def\N{{\cal N}}

\def\ka{{\kappa}}

\newfont{\Blackboard}{msbm10 scaled 1200}

\newfont{\roma}{cmr10 scaled 1200}

\def\D{{\cal D}}
\def\II{{\cal I}}

\def\<{{\langle}}
\def\>{{\rangle}}

\def\Ga{\Gamma}
\def\var{\varphi}
\def\si{\sigma}

\def\a{\alpha}
\def\b{\beta}

\def\Om{\Omega}

\newtheorem{thm}{{}\hskip\parindent Theorem}[section]
\newtheorem{lem}{{}\hskip\parindent Lemma}[section]
\newtheorem{pro}{{}\hskip\parindent
Proposition}[section]

\newtheorem{dfn}{{}\hskip\parindent
Definition}[section]

\def\pl{\partial}
\def\rw{\rightarrow}

\def\be{\begin{equation}}
\def\ee{\end{equation}}
\def\beq{\arraycolsep=1.5pt\begin{eqnarray}}
\def\eeq{\end{eqnarray}}

\def\P{\cal P}
\def\R{I\!\!R}

\def\n{\vec{n}}
\title{Infinitesimal Rigidity of Strain Tensors for Shells with Mixed Type and its Applications }
\date{}
\author{
Liang-Biao Chen$^a$ and Peng-Fei Yao$^{ab}$\\[0.2cm]
\nonumber
a\,\,
Key Laboratory of Systems and Control\\\nonumber
Institute of Systems Science,
Academy of Mathematics and Systems Science\\\nonumber
Chinese Academy of Sciences, Beijing 100190, P. R.
China\\\nonumber
b\,\,
School of Mathematical Sciences\\\nonumber
 Shanxi University, Taiyuan 030006, China\\\nonumber
e-mail: pfyao@iss.ac.cn}

\begin{document}
\maketitle
\footnote{This work is supported by the National
Science Foundation of China, grants no. 12071463.}

\begin{quote}
\begin{small}
{\bf Abstract} \,\,\,We derive an infinitesimal  rigidity lemma for the strain tensor of surfaces with their curvatures changing sign. As an application, we obtain the optimal
constant in the first Korn inequality scales like $h^{4/3}$ for such shells of mixed type.
\end{small}
\end{quote}

\section{Introduction and the Main Results}
\def\theequation{1.\arabic{equation}}
\hskip\parindent The goal of the present paper is twofold to study the rigidity  of the strain tensors of surfaces and then to apply it to obtain the optimal constant in the first Korn inequality  for some shells of mixed type.

The linear strain equations plays a fundamental role in the theory of thin shells, see \cite{HoLePa, LePa,LeMoPa, LeMoPa1, Yao2017}. The solvability of the strain equation is needed to prove the density of smooth infinitesimal isometries  in the $W^{2,2}(\Om,\R^3)$ infinitesimal isometries and  the matching property of the smooth enough infinitesimal isometries  with higher
order infinitesimal isometries \cite{HoLePa, LeMoPa1, Yao2017}. This matching property is an important tool in
obtaining recovery sequences ($\Ga$-lim sup inequality) for dimensionally-reduced shell theories
in elasticity, when the elastic energy density scales like $h^\b,$ $\b\in(2, 4),$ that is, the intermediate
regime between pure bending ($\b=2$) and the von-Karman regime ($\b=4$). Such results have been obtained for elliptic surfaces \cite{LeMoPa1}, the developable surface \cite{HoLePa}, the  hyperbolic surface  \cite{Yao2017}, the degenerated surface \cite{CY}, and the surface with the Gaussian curvature changing sign \cite{CY1}, respectively.
A survey on this topic is presented in \cite{LePa}.

The type of strain tensors depends on the sign of the curvature of the surface: It is elliptic if the curvature is positive; it is parabolic if the curvature is zero but $\Pi\not=0;$  it is hyperbolic for the negative curvature. When the curvature of the surface  changes its sign, the strain tensor is of mixed type.

Moreover, the lower regularity of strain tensors of surfaces represents the rigidity of geometrical shape of surfaces.

Linear strain equations have  been studied in \cite{CY1} for the surface with its curvature changing sign, where the density of smooth infinitesimal isometries in the $W^{2,2}$ infinitesimal isometries is obtained and the  matching property of infinitesimal
isometries  is proved. In particularly, \cite{CY1} established the following regularity
$$\|y\|^2_{\WW^{1,2}(S,\R^3)}\leq C\|U\|^2_{\WW^{2,2}(S,T^2_{\sym})},$$ where $y$ is a solution to problem (\ref{s}) later. However, the above regularity is not enough for establishing the rigidity results of the surface. Here, we derive the $\LL^2$ regularity, that is an infinitesimal  rigidity lemma (Theorem \ref{t1.1}) for the strain tensors.  This lemma  is one of the key ingredients for the optimal constant in the first Korn inequality  for  shells. The basic tools are  some results on the generalized tensors.

Originally, Korn's inequalities were used to prove existence, uniqueness and well-posedness of
boundary value problems of linear elasticity \cite{Ci,Lo}. The optimal exponential of thickness in Korn's inequalities for thin shells represents the relationship between the rigidity and the thickness of a shell when the small deformations take place since  Korn's inequalities are linearized from the geometric rigidity inequalities under the small deformations \cite{FrJaMu}. Thus it is the best Korn constant in the Korn inequality that is of
central importance  \cite{CiOlTr,LeMu, Na, Na1, PaTo, PaTo1}. Moreover, one has known that the best Korn constant subject to the Gaussian curvature. Under the assumption that the middle surface of the shell is given by a single principal coordinate system, the one for the parabolic shell  scales like $h^{3/2}$ \cite{GH,GH1}, for the hyperbolic shell, $h^{4/3}$ \cite{Ha2} and for the elliptic shell, $h$ \cite{Ha2}. Later the assumption of the single principal coordinate  is removed in \cite{Yao2018,Yao2019}, where the case of the closed elliptic shell is particularly included.

To the best knowledge of the authors, the present paper for the first time establishes the rigidity results for the surface with the curvature changing sign (Theorems \ref{t1.1} and \ref{t1.2}) and thus
yields that the optimal constant in the Korn inequality  scales like $h^{4/3}$ for some shells of mixed type (Theorem \ref{t1.3}).

Let $S\subset M$ be given by
\be S=\{\,\a(t,s)\,|\,(t,s)\in[0,a)\times(-b_0,b_1)\,\},\quad a>0,\quad b_0>0,\quad b>0,\label{Q0}\ee where $\a:$ $[0,a)\times[0,b]\rw M$ is an imbedding map which is a family of regular curves with two parameters $t,$ $s$ such that
\be \Pi(\a_t(t,s),\a_t(t,s))>0,\quad\mbox{for all}\quad (t,s)\in[0,a)\times[-b_0,0],\label{Pi2.801}\ee
where $\a(\cdot,s)$ is a closed curve with  period $a$ for each $s\in[-b_0,b_1].$  Set
$$S=S^+\cup\Ga_0\cup S^-,$$ where
$$S^+=\{\,\a(t,s)\,|\,(t,s)\in[0,a)\times(0,b_1)\,\},\quad \Ga_0=\{\,\a(t,0)\,|\,t\in[0,a]\,\},$$
$$ S^-=\{\,\a(t,s)\,|\,(t,s)\in[0,a)\times(-b_0,0)\,\}.$$

 {\bf Curvature assumptions}\,\,\, Let $\kappa$ be the Gaussian curvature function on $M.$ We assume that $S$ satisfies the following curvature conditions:
\be\kappa(x)>0\qfq x\in S^+\cup\Ga_{b_1};\label{kappa+}\ee
\be\kappa=0,\quad D\kappa(x)\not=0\qfq x\in\Ga_0;\label{kappa0}\ee
\be\kappa(x)<0\qfq x\in S^-\cup\Ga_{-b_0},\label{kappa-}\ee where
$$\Ga_{b_1}=\{\,\a(t,b_1)\,|\,t\in[0,a)\,\},\quad \Ga_{-b_0}=\{\,\a(t,-b_0)\,|\,t\in[0,a)\,\}.$$

Let $y\in \WW^{1,2}(S,\R^3)$ be a displacement of the middle surface $S.$ We decompose $y$ as
$$y=W+w\n,\quad w=\<y,\n\>,$$ where $\<\cdot,\cdot\>$ denotes the dot metric of the Euclidean space $\R^3.$ The (linear) strain tensor of the middle surface (related to the displacement $y$) is defined by $$\Upsilon(y)=\sym DW+w\Pi,$$ where $D$ is the Levi-Civita connection of the induced metric $g$ on $S,$
$$\sym DW=\frac12(DW+DW^T),$$ and
$$\Pi(\a,\b)=\<\nabla_\a\n,\b\>\qfq\a,\,\,\b \in T_xM,\quad x\in S$$ is the second fundamental form of surface $M.$
For given $U\in L^2(S,T^2)$ , consider problem
\be\Upsilon(y)=U\qfq p\in S.\label{s}\ee

Let $T^k$ denote the all $k$th-order tensor fields on $S.$ Let $L^2(S,T^k)$ be the space of all $k$th-order tensor fields on $S$ with the norm
$$(P_1,P_2)_{\LL^2(S,T^k)}=\int_S\<P_1,P_2\>dg,$$ where
$$\<P_1,P_2\>=\sum_{i_1,\cdots,i_k=1}^2P_1(e_{i_1},\cdots,e_{i_k})P_2(e_{i_1},\cdots,e_{i_k})\qfq x\in S,$$ and $e_1,$ $e_2$ is an orthonormal basis of $T_xS.$
Let $[\WW^{m,2}(S,T^k)]'$ be the dual spaces, based on the inner product $(\cdot,\cdot)_{\LL^2(S,T^k)}$ of $\LL^2(S,T^k).$

Our main results are the following, that is, an infinitesimal rigidity lemma on the strain tensor of the middle surface with mixed type.

\begin{thm}\label{t1.1} Suppose that region $S$ of $\CC^5$ satisfies $(\ref{Q0})$-$(\ref{kappa-}).$  Let $y=W+w\n$ solve problem  $(\ref{s}).$ Then there is a constant $C>0,$ independent of $y,$ such that
\be\|W\|_{L^2(S,T)}^2+\|w\|^2_{[\WW^{1,2}(S)]'}\leq C(\|U\|^2_{L^2(S,T^2)}+\|W\|^2_{\LL^2(\pl S,T)}).\label{1.3}\ee
\end{thm}

The proof of Theorem \ref{t1.1}  will be presented in the end of Section 3.

For the hyperbolic surface $S,$ estimate $(\ref{1.3})$ is given in \cite{Yao2019}. If $S$ is elliptic, we then have
$$\|DW\|_{L^2(S,T^2)}^2+\|w\|^2_{\LL^2(S)}\leq C(\|U\|^2_{L^2(S,T^2)}+\|W\|^2_{\LL^2(\pl S,T)}),$$ see \cite{Ci,Yao2018}.

Using Theorem \ref{t1.1} and by the same argument as in \cite[{\bf Proof of Theorem 1.3}]{Yao2019}, we have the following.

\begin{thm}\label{t1.2} Suppose that region $S$ of $\CC^5$ satisfies $(\ref{Q0})$-$(\ref{kappa-}).$ Then
\be \|w\|^2_{L^2(S)}\leq C(\|Dw\|_{L^2(S,T)}\|\Upsilon(y)\|_{L^2(S,T^2)}+\|\Upsilon(y)\|^2_{L^2(S,T^2)})\label{x1.12}\ee
for all $y=W+w\n\in \WW^{1,2}(S,\R^3)$
with $y|_{\Ga_{-b_0}\cup\Ga_{b_1}}=0.$
\end{thm}

{\bf Application to elasticity of thin shells.}\,\,\, Suppose that $S$ is the middle surface of the shell with thickness $h>0$
$$\Om=\{\,x+t\n(x)\,|\,x\in S,\,-h/2<t<h/2\,\}.$$
By the same arguments as in \cite[{\bf Proof of Theorem 1.4}]{Yao2019} we have the following results of the optimal exponential of  the first Korn inequality for shells of mixed type.  The details of the proofs are omitted.

\begin{thm}\label{t1.3} Suppose that region $S$ of $\CC^5$ satisfies $(\ref{Q0})$-$(\ref{kappa-}).$
There are $C>0,$ $h_0>0,$ independent of $h>0,$ such that
\be\|\nabla y\|^2_{L^2(\Om,\R^{3\times3})}\leq\frac{C}{h^{4/3}}\|\sym\nabla y\|^2_{L^2(\Om,\R^{3\times3})}\label{1.13}\ee
for all $h\in(0,h_0)$ and $y\in \WW^{1,2}(S,\R^3)$ with $y|_{\Sigma}=0,$ where
$$\Sigma=\{\,x+t\n\,|\,x\in\Ga_{-b_0}\cup\Ga_{b_1},\,\,-h/2<t<h/2\,\}.$$ Moreover, the exponential of the thickness $h$ in $(\ref{1.13})$ is optimal.
\end{thm}

Estimates (\ref{x1.12}) and (\ref{1.13}) have been   given in \cite{Ha2} when the middle surface is given by one single principal coordinate, where $S$ is a hyperbolic surface. The assumption that
$S$ consists of one single principal coordinate is removed in \cite{Yao2019}.

\setcounter{equation}{0}
\section{Generalized Tensors}
\def\theequation{2.\arabic{equation}}
\hskip\parindent Let $m\geq0$ and $k\geq0$ be  integers. Denote by $\CC_0^m(S,T^k)$  all the $k$ order tensor fields with  continuously mth order derivatives and compact supports on $\ol{S}.$
Denote by $\D_m(S,T^k)$  the  space $\CC^m_0(S,T^k)$ with the following convergence: a sequence $\{\var_n\}\subset\CC_0^m(S,T^k)$ is said to converge to $\var_0\in\CC_0^m(S,T^k)$ if

(1) there exists a subset
$K\subset\subset S$ such that
$$\supp \var_n\subset K\qflq n\geq1;$$

(2) $\sup_{x\in K}|D_{X_1}\cdots D_{X_j}(\var_n-\var_0)|\rw0$ as $n\rw\infty$ for all $X_1,$ $\cdots,$ $X_j\in \CC^m(S,T)$ where $0\leq j\leq m.$

Let $\D'_m(S,T^k)$ consist of all continuously linear functionals on $\D_m(S,T^k).$  An element in $\D_m'(S,T^k)$ is said to be {\it a generalized tensor  of  order $k$.}

We define some generalized derivative operations as follows.
\begin{dfn}
For given $w\in\D_m'(S),$ the generalized gradient $Dw\in\D_{m+1}'(S,T)$ is defined by
\be Dw(X)=-w(\div X)\qflq X\in \D_{m+1}(S,T).\label{n1.2}\ee

For given $W\in\D_m'(S,T)$ and $y\in\D_m'(S,\R^3),$ the generalized differentials  $DW\in\D_{m+1}'(S,T^2)$ and $\nabla y\in\D_{m+1}'(S,T^2)$ are defined by
\be DW(P)=-W(\div P)\qflq P\in\D_{m+1}(S,T^2),\label{x1.4}\ee and
\be \nabla y(P)=y(-\div P+\<P,\Pi\>\n)\qflq P\in \D_{m+1}(S,T^2),\label{n1.3}\ee respectively.

For given $W\in\D_m'(S,T)$ and $U\in\D_m'(S,T^2),$ the generalized divergences $\div W\in\D'_{m+1}(S)$ and $\div U\in\D_{m+1}'(S,T)$ are defined by
\be \div W(z)=-W(Dz)\qflq z\in\D_{m+1}(S),\label{n1.4}\ee and
\be \div U(Z)=-U(DZ)\qflq Z\in\D_{m+1}(S,T),\label{1.5}\ee respectively.

For given $y\in\D_m'(S,\R^3)$ and $X\in\D_{m+1}(S,T),$ we define $\nabla_Xy\in\D'_{m+1}(S,T)$ by
\be\nabla_Xy(Y)=\nabla y(Y\otimes X)\qflq Y\in\D_{m+1}(S,T).\label{1.6}\ee
$\nabla_Xy\in\D_{m+1}'(S,T)$ is said to be the generalized derivative of $y$ along direction $X.$ Similarly, for given $W\in\D_m'(S,T)$ and $X\in\D_{m+1}(S,T),$ $D_XW\in\D_{m+1}'(S,T)$ is defined by
\be D_XW(Y)=DW(Y\otimes X)\qflq Y\in\D_{m+1}(S,T).\label{1.7}\ee
\end{dfn}

In the above definition the right hand sides of formulas (\ref{n1.2})-(\ref{1.7}) are linear functionals, respectively. For instant,
if $W\in\LL^2(S,T),$ then we view $W\in\D_m'(S,T)$ as
$$W(Y)=\int_S\<W,Y\>dg\qflq Y\in\D_m(S,T),$$
and so on.

\begin{pro}
For given $y\in\D_m'(S,\R^3)$ and $X\in\D_{m+1}(S,T),$
\be\nabla_Xy(Y)=-y\Big((\div X) Y+\nabla_XY\Big)\qfq Y\in\D_{m+1}(S,T).\label{1.8}\ee
If $y\in\WW^{1,2}(S,\R^3),$ then
\be\nabla_Xy(Y)=\int_S\<\nabla_Xy,Y\>dg\qfq Y\in\D_{m+1}(S,T).\ee
Similarly, if  $W\in \WW^{1,2}(S,T)$ and $X\in\D_{m+1}(S,T),$ then
\be D_XW(Y)=\int_S\<D_XW,Y\>dg\qflq Y\in\D_{m+1}(S,T),\label{n1.10}\ee where $D_XW$ is defined by $(\ref{1.7}).$
\end{pro}

\begin{proof} Let $x\in S$ be given.  Suppose that $E_1,$ $E_2$ is a frame normal at $x.$ Then
\beq\div(Y\otimes X)&&=\sum_{ij=1}^2D(Y\otimes X)(E_i,E_j,E_j)E_i=\sum_{ij=1}^2E_j(\<Y,E_i\>\<X,E_j\>)E_i\nonumber\\
&&=\sum_{ij=1}^2(\<D_{E_j}Y,E_i\>\<X,E_j\>+\<Y,E_i\>\<D_{E_j}X,E_j\>)E_i\nonumber\\
&&=D_XY+(\div X)Y\qaq x\qfq X,\,Y\in\D_{m+1}(S,T).\label{1.10}\eeq
Thus (\ref{1.8}) follows from definition (\ref{1.6}).

Let $y\in\WW^{1,2}(S,\R^3).$ Then $y\in\D_m'(S,\R^3)$ is viewed as a linear functional on $\D_m(S,T)$ by
$$y(Y)=\int_S\<y,Y\>dg\qflq Y\in\D_m(S,T).$$ Noting that
$$\nabla_XY=D_XY-\<Y\otimes X,\Pi\>\n,$$
from (\ref{1.6}), (\ref{n1.3}), and (\ref{1.10}), we have
\beq \nabla_Xy(Y)&&=y\Big(-\div(Y\otimes X)+\<Y\otimes X,\Pi\>\n\Big)=\int_S\<y,\,-\div(Y\otimes X)+\<Y\otimes X,\Pi\>\n\>dg\nonumber\\
&&=\int_S\<y,-\nabla_XY-(\div X)Y\>dg=-\int_\Ga\<y,Y\>\<X,\nu\>d\Ga+\int_S\<\nabla_Xy,Y\>dg\nonumber\\
&&=\int_S\<\nabla_Xy,Y\>dg\qfq Y\in\D_{m+1}(S,T).\nonumber\eeq A similar argument yields (\ref{n1.10}).
\end{proof}

Let $\E_m(S,T^k)$ be the space $\CC^m(S,T^k)$ with the convergence: A sequence $\{\var_n\}\subset\CC^m(S,T^k)$ is said to converge to some $\var_0\in\CC^m(S,T^k)$ if for any given compact set
$K\subset\subset S$ and any given $X_1,$ $\cdots,$ $X_j\in\CC^m(S,T)$ where $0\leq j\leq m,$ there holds
$$\lim_{n\rw0}\sup_{x\in K}|D_{X_1}\cdots D_{X_j}(\var_n-\var_0)|=0.$$
Denote by $\E_m'(S,T^k)$ all the continuous linear functionals on $\E_m(S,T^k).$

 We define some multiplier operations between generalized tensor  as follows.
\begin{dfn} For given $w\in\D_m'(S)$ and $f\in\E_m(S),$ $fw\in\D_m'(S)$ is defined by
\be (fw)(z)=w(fz)\qflq z\in\D_m(S);\label{1.12}\ee
for given $W\in\D_m'(S,T)$ and $f\in\E_m(S),$ $fW\in\D_m'(S,T)$ by
\be (fW)(Z)=W(fZ)\qflq Z\in\D_m(S,T);\ee
for given $U\in\D_m'(S,T^2)$ and $f\in\E_m(S),$ $fU\in\D_m'(S,T^2)$ by
\be (fU)(P)=U(fP)\qflq P\in\D_m(S,T^2);\ee
for given $W\in\D_m'(S,T)$ and $F\in\E_m(S,T),$ $(\<W,F\>)\in\D_m'(S)$ by
\be (\<W,F\>)(z)=W(zF)\qflq z\in\D_m(S);\label{n2.15}\ee
for given $w\in\D_m'(S)$ and $X\in\E_m(S,T),$ $wX\in\D_m'(S,T)$ by
\be (wX)(Z)=w(\<X,Z\>)\qflq Z\in\D_m(S,T);\label{1.16}\ee
for given $U\in\D_m'(S,T^2)$ and $R\in\E_m(S,T^2),$ $\<U,R\>\in\D_m'(S)$ by
\be (\<U,R\>)(z)=U(zR)\qflq z\in\D_m(S);\label{n1.17}\ee
for given $w\in\D_m'(S)$ and $R\in\E_m(S,T^2),$ $wR\in\D_m'(S,T^2)$  by
\be (wR)(P)=w(\<R,P\>)\qflq P\in\D_m(S,T^2);\label{1.17}\ee
for given $U\in\D_m'(S,T^2),$ $\tr U\in\D_m'(S)$  by
\be \tr U(z)=U(z\id)\qflq z\in\D_m(S);\label{1.19}\ee
for given $W\in\D_m'(S,T)$ and $R\in\E_m(S,T^2),$ $RW\in\D_m'(S,T)$  by
\be (RW)(Z)=W(R^TZ)\qflq Z\in\D_m(S,T);\label{1.19}\ee
for given $U\in\D_m'(S,T^2)$ and $R_1,$ $R_2\in\E_m(S,T^2),$ $R_1UR_2\in\D_m'(S,T^2)$  by
\be (R_1UR_2)(P)=U(R_1^TPR_2^T)\qflq P\in\D_m(S,P).\ee
\end{dfn}

We need a linear operator $Q$ (\cite{CY, CY1, Yao2017,Yao2019}) as follows. For each point
$p\in M$, the Riesz representation theorem implies that there
exists an isomorphism $Q:$ $T_pM\rw T_pM$ such that
\be \label{n2.31}
\<\a,Q\b\>=\det\left(\a,\b,\vec{n}(p)\right)\qfq\a,\,\b\in T_pM.\ee
Let $e_1,$ $e_2$ be an orthonormal basis of $T_pM$ with
positive orientation, that is,
$$\det\Big(e_1,e_2,\n(p)\Big)=1.$$
Then $Q$ can be expressed explicitly by
\be Q\a=\<\a,e_2\>e_1-\<\a,e_1\>e_2\quad\mbox{for
all}\quad\a\in T_pM.\label{n2.32}\ee
Clearly, $Q$ satisfies
$$ Q^T=-Q,\quad Q^2=-\Id.$$
Operator $Q$ plays an important role in our analysis.

F given or $U\in\D_m'(S,T),$ we define $QU\in\D_m'(S,T^2)$ by
$$(QU)(X\otimes Y)=U(X\otimes Q^TY),\quad (UQ)(X\otimes Y)=U(QX\otimes Y)\qfq X,\,Y\in\D_m(S,T).$$ Moreover, for given $W\in \D'_m(S,T),$ we define $QW\in\D'_m(S,T)$ by
\be (QW)(F)=-W(QF)\qfq F\in D_m(S,T).\label{n2x.24}\ee

\begin{pro}The following formulas hold:
\be \div(wR)=R^TDw+w\div R\qfq w\in\D_m'(S),\quad R\in\E_m(S,T^2),\label{1.22}\ee
\be \div(RW)=\<R,DW\>+\<\div R,W\>\qfq W\in\D_m'(S,T),\quad R\in\E_m(S,T^2),\label{1.23}\ee
\be UQ+QU^T=(\tr U)Q\qfq U\in\D_m'(S,T^2).\label{1.24}\ee
\end{pro}

\begin{proof} Let $x\in S$ be given.  Suppose that $E_1,$ $E_2$ is a frame normal at $x.$ Then
\beq\<R,DZ\>&&=\sum_{ij}R(E_i,E_j)\<D_{E_j}Z,E_i\>\nonumber\\
&&=\sum_{ij}\{E_j[R(E_i,E_j)\<Z,E_i\>]-DR(E_i,E_j,E_j)\<Z,E_i\>\}\nonumber\\
&&=\sum_{j}\{E_j[RZ(E_j)]-\ii_ZDR(E_j,E_j)\}\nonumber\\
&&=\div(RZ)-\<\div R,Z\>\qaq x.\nonumber\eeq
It follows from (\ref{1.5}), (\ref{1.19}), (\ref{1.16}), and (\ref{1.17}) that
\beq\div(wR)(Z)&&=-(wR)(DZ)=-w(\<R,DZ\>)=-w\Big(\div(RZ)\Big)+w(\<\div R,Z\>)\nonumber\\
&&=Dw(RZ)+(w\div R)(Z)=(R^TDw+w\div R)(Z)\nonumber\eeq for $Z\in\E_{m+1}(S,T).$

For given $z\in\D_{m+1}(S),$ we have at $x$
\beq R^TDz&&=\sum_{ij}E_i(z)R(E_j,E_i)E_j=\sum_{ij}\{E_i[zR(E_j,E_i)]-zDR(E_j,E_i,E_i)\}E_j\nonumber\\
&&=\sum_{j}\{\<\div(zR),E_j\>-z\tr\ii_{E_j}DR\}E_j\nonumber\\
&&=\div(zR)-z\div R.\label{x2.27}\eeq
By (\ref{n1.4}), (\ref{1.19}), (\ref{x2.27}), (\ref{x1.4}), and (\ref{n1.17}),
\beq\div(RW)(z)&&=-(RW)(Dz)=-W(R^TDz)\nonumber\\
&&=DW(zR)+W(z\div R)\nonumber\\
&&=(\<DW,R\>)(z)+(\<W,\div R\>)(z)\qfq z\in D_{m+1}(S).\nonumber\eeq

Finally, we shall prove (\ref{1.24}). Let $E_1,$ $E_2$ be a local frame with positive orientation. Then
$$QE_1=-E_2,\quad QE_2=E_1,\quad Q^TE_1=E_2,\quad Q^TE_2=-E_1.$$
Then
$$(UQ+QU^T)(E_1\otimes E_2)=U(QE_1\otimes E_2)+U( Q^TE_2\otimes E_1)=(\tr U)\<QE_1,E_2\>,$$
$$(UQ+QU^T)(E_2\otimes E_1)=(\tr U)\<QE_2,E_1\>,\quad(UQ+QU^T)(E_i\otimes E_i)=(\tr U)\<QE_i,E_i\> $$ for $1\leq i\leq2.$
\end{proof}

Let $S$ be of $\CC^{m+1}.$ Let $y\in D'_{m}(S,\mathbb{R}^3)$ be given. We define $w\in\D_m'(S),$ $W\in\D_m(S,T),$ $v\in\D_{m+1}(S),$ $V\in\D'_{m+1}(S,T),$ and $U\in\D'_{m+1}(S,T^2)$ by
\be w(z)=y(z\n)\qfq z\in\D_{m}(S),\quad W(Z)=y(Z)\qfq Z\in\D_{m}(S,T),\label{2.27}\ee
\be  v(z)=\frac12y(QDz)\qfq z\in\D_{m+1}(S),\label{2.28}\ee
\be V(Z)=y\Big((\div QZ)\n+\nabla\n QZ\Big)\qfq Z\in\D_{m+1}(S,T),\label{V2.29}\ee and
\be U(P)=\nabla y(\sym P)\qfq P\in\D_{m+1}(S,T^2),\label{U2.31}\ee respectively. Then
\be\sym DW+w\nabla\n=U\qiq\D'_{m+1}(S,T^2)\label{n1.25}\ee in the sense of generalized tensors.

\begin{thm}\label{t2.1}For given $y\in D'_m(S,\mathbb{R}^3),$ we have
\be v=\frac12\div QW\qfq x\in S,\label{n2x.23}\ee
\be\left\{\begin{array}{l}Dw=\nabla\n W-QV\qfq x\in S,\\
\div W=-(\tr\Pi)w+\tr U\qfq x\in S,\end{array}\right.\label{1.25}\ee
and
\be\left\{\begin{array}{l}Dv=\nabla\n V+Q\div QUQ\qfq x\in S,\\
\div V=-(\tr\Pi)v-\<Q\nabla\n,U\>\qfq x\in S,\end{array}\right.\label{1.26}\ee in the sense of generalized tensors.
\end{thm}

\begin{proof}
By (\ref{n2x.24}) and (\ref{n1.4}) we have
$$ 2v(z)=(W+w\n)(QDz)=-(QW)(Dz)=(\div QW)(z)\qfq z\in D_{m+1}(S).$$

For $Z\in\D_{m+1}(S,T)$ and by (\ref{n1.2}) and (\ref{1.19})
\beq(Dw-\nabla\n W)(Z)&&=-w(\div Z)-W(\nabla^T\n Z)=-y\Big((\div Z)\n+\nabla\n Z\Big)\nonumber\\
&&=-V(Q^TZ)=-(QV)(Z).\nonumber\eeq

By (\ref{n1.3}), (\ref{x1.4}), and (\ref{1.17})
$$\nabla y(P)=y(-\div P+\<P,\Pi\>\n)=-W(\div P)+w(\<P,\Pi\>)=(DW+w\Pi)(P)$$
for any $P\in\D_{m+1}(S,T^2).$ Then
\be\nabla y=DW+w\Pi\qiq \D_m'(S,T^2).\label{1.27}\ee Since
$$\div(z\id)=Dz,$$ from (\ref{1.17}), (\ref{1.12}), and (\ref{1.19}), we obtain
\beq\div W(z)&&=-W(Dz)=-W\Big(\div(z\id)\Big)=DW(z\id)\nonumber\\
&&=\nabla y(z\id)-(w\Pi)(z\id)=\nabla y\Big(\sym(z\id)\Big)-w\Big((\tr\Pi)z\Big)\nonumber\\
&&=U(z\id)-(\tr\Pi)w(z)=(\tr U)(z)-(\tr \Pi)w(z)\qfq z\in\D_{m+1}(S),\nonumber\eeq that is, the second equation in (\ref{1.25}).

Next, we shall prove (\ref{1.26}). Let $x\in S$ be given.  Suppose that $E_1,$ $E_2$ is a frame normal at $x$ with positive orientation. Then $E_1=QE_2.$ Let $R$ be the curvature tensor. By the Ricci identity
$$D^2Z(X_1,X_2,X_3)=D^2Z(X_1,X_2,X_3)+R(X_2,X_3,Z,X_1)$$ for $Z,$ $X_1,$ $X_2,$ $X_3\in\D_{m+1}(S,T).$
Using the above formula, we have
\beq \<\div\sym DZ,E_2\>&&=E_1(\sym DZ)(E_2,E_1)+E_2(\sym DZ)(E_2,E_2)\nonumber\\
&&=\frac12E_1[DZ(E_2,E_1)+DZ(E_1,E_2)]+E_2[DZ(E_2,E_2)]\nonumber\\
&&=\frac12[D^2Z(E_2,E_1,E_1)+D^2Z(E_1,E_2,E_1)]+E_2(\div Z)-D^2Z(E_1,E_1,E_2)\nonumber\\
&&=E_2(\div Z)+\frac12[D^2Z(E_2,E_1,E_1)-D^2Z(E_1,E_2,E_1)]+\kappa\<Z,E_2\>\nonumber\\
&&=E_2(\div Z)+\kappa\<Z,E_2\>+\frac12E_1(\<D_{E_1}Z, E_2\>-\<D_{E_2}Z,E_1\>)\nonumber\\
&&=E_2(\div Z)+\kappa\<Z,E_2\>+\frac12\<D\div QZ,E_1\>\nonumber\\
&&=\<D\div Z+\kappa Z-\frac12QD\div QZ,\,\,E_2\>\qfq Z\in\D_{m+2}(S,T).\nonumber\eeq
A similar computation yields
$$\<\div\sym DZ,E_1\>=\<D\div Z+\kappa Z-\frac12QD\div QZ,\,\,E_1\>\qaq x.$$ Thus
\be\div\sym DZ=D\div Z+\kappa Z-\frac12QD\div QZ\qfq Z\in\D_{m+2}(S,T).\label{1.28}\ee

By (\ref{1.5}) and (\ref{1.28})
\beq(\div \sym DW)(Z)&&=-DW(\sym DZ)=W(\div\sym DZ)\nonumber\\
&&=W(D\div Z+\kappa Z-\frac12QD\div QZ)\nonumber\\
&&=(D\div W+\kappa W)(Z)-\frac12y(QD\div QZ)\nonumber\\
&&=(D\div W+\kappa W)(Z)-v(\div QZ)\nonumber\\
&&=(D\div W+\kappa W-QDv)(Z)\qfq Z\in\D_{m+2}(S,T).\nonumber\eeq
Using the relations
$$\nabla\n Q\nabla\n=\ka Q,\quad \nabla\n-(\tr \Pi) \id=Q\nabla\n Q,\quad \div\nabla\n=D(\tr\Pi),$$ we have, by (\ref{1.25}), (\ref{1.22}), and (\ref{n1.25}),
\beq \nabla\n V(Z)&&=V(\nabla\n Z)=QV(Q\nabla\n Z)=(\nabla\n W-Dw)(Q\nabla\n Z)\nonumber\\
&&=W(\nabla\n Q\nabla\n Z)-Dw\Big((\tr\Pi)QZ-\nabla\n QZ\Big)\nonumber\\
&&=[\kappa W+\nabla\n Dw-(\tr\Pi)Dw](QZ)\nonumber\\
&&=[\kappa W+\div(w\nabla\n)-w\div\nabla\n-(\tr\Pi)Dw](QZ)\nonumber\\
&&=[\kappa W+\div U-\div\sym DW-D((\tr\Pi)w)](QZ)\nonumber\\
&&=[\div U-D\div W+QDv-D((\tr\Pi)w)](QZ)\nonumber\\
&&=[Dv-Q\div U+QD(\tr U)](Z)\nonumber\\
&&=[Dv-Q\div\Big(U-(\tr U)\id\Big)](Z)\nonumber\\
&&=[Dv-Q\div QUQ](Z)\qfq Z\in\D_m(S,T),\nonumber\eeq where the formula $U-(\tr U)\id=QUQ$ has been used, that is the first equation in (\ref{1.26}).

Noting that
$$\<\Pi,Q\nabla\n\>=\<Q,\Pi\>=0,\quad\div QDz=0,\quad \div(zQ\nabla\n)=-\nabla\n QDz,$$
$$\div\Big(z(\tr\Pi)Q\Big)=-QD(z\tr\Pi)\qfq z\in\D_{m+1}(S),$$ we compute, by (\ref{1.27}), (\ref{n1.3}), (\ref{n1.17}), and (\ref{1.23}),
\beq \div V(z)&&=-V(Dz)=-y\Big((\div QDz)\n+\nabla\n QDz\Big)=-W(\nabla\n QDz)\nonumber\\
&&=W\Big(\div(zQ\nabla\n)\Big)=-DW(zQ\nabla\n)=-\nabla y(zQ\nabla\n)+w\Pi(zQ\nabla\n)\nonumber\\
&&=-\nabla y(z\sym Q\nabla\n)+\frac12\nabla y(z\nabla\n Q^T-zQ\nabla\n)+w(z\<\Pi,Q\nabla\n\>)\nonumber\\
&&=-U(zQ\nabla\n)+\frac12\nabla y(z\nabla\n Q^T-zQ\nabla\n)\nonumber\\
&&=-U(zQ\nabla\n)-\frac12\nabla y\Big(z(\tr\Pi)Q\Big)\nonumber\\
&&=-(\<U,Q\nabla\n\>)(z)-\frac12y\Big(-\div(z(\tr\Pi)Q)+z(\tr\Pi)\<Q,\Pi\>\n\Big)\nonumber\\
&&=-(\<U,Q\nabla\n\>)(z)-v(z\tr\Pi)=-(\<U,Q\nabla\n\>)(z)-((\tr\Pi)v)(z)\nonumber\eeq for $z\in\D_{m+1}(Z),$ where the formula $Q\nabla\n=(Q\nabla\n)^T+(\tr\Pi)Q$ has been used.
Thus the second equation in (\ref{1.26}) is true.
\end{proof}

\setcounter{equation}{0}
\section{Proof of Theorem \ref{t1.1}}
\def\theequation{3.\arabic{equation}}
\hskip\parindent  Let $\Om\subset M$ be  a bounded Lipschitz region with boundary $\Ga.$ Consider operator $\B:$ $\LL^2(\Om,T)\rw\LL^2(\Om,T)$ given by
$$\B W=\div DW,\quad D(\B)=\WW^{2,2}(\Om,T)\cap\WW^{1,2}_0(\Om,T).$$
Let $W\in D(\B)$ and $x\in\Om$ be fixed. Let $x\in\Om$ be fixed. Let $E_1,$ $E_2$ be a frame field normal at $x.$ By \cite[Theorem 1.26]{Yao2011}
$$\<\B W,E_i\>=\sum_{j=1}^2D(DW)(E_i,E_j,E_j)=\<\sum_jD_{E_j}^2W, E_i\>=-\<{\bf \Delta}W,E_i\>+\Ric(W,E_i)$$ at $x$ for $i=1,$  $2,$
where ${\bf\Delta}$ is the Hodge-Laplacian and $\Ric$ is the Ricci tensor. Thus
$$\B W=-{\bf \Delta}W+\ii_W\Ric.$$ By \cite{Ta} $\B$ is a self-adjoint operator with the compact resolvent. We then have a direct sum
$$\LL^2(\Om,T)=\R(\B)\oplus \N(\B),$$ where $\R(\B)$ and $\N(\B)$ are the valued field and the null space of $\B,$ respectively. Moreover, $\dim\N<\infty,$ and for any $Y\in\R(\B),$ there exists a unique
$Z\in D(\B)\cap\R(\B)$ that solves problem
$$\div DZ=Y,\quad Z|_\Ga=0$$ and satisfies
$$\|Z\|_{\WW^{2,2}(\Om,T)}\leq C\|Y\|_{\LL^2(\Om,T)}.$$

\begin{lem}\label{l3.1} There exists constant $C>0$ such that
\be\|W\|^2_{\LL^2(\Om,T)}\leq C(\|DW\|^2_{\WW^{-1,2}(\Om,T^2)}+\|W\|^2_{\LL^2(\Ga,T)})\qflq W\in\LL^2(\Om,T).\label{3.29}\ee
\end{lem}

\begin{proof}
{\bf Step 1.}\,\,\,First, we claim that there is a constant $C>0$ such that
\be\|W\|^2_{\LL^2(\Om,T)}\leq C(\|DW\|^2_{\WW^{-1,2}(\Om,T^2)}+\|W\|^2_{\LL^2(\Ga,T)})\qflq W\in\R(\B).\label{3.30}\ee Let $W\in\R(\B)$ be given. For any $Y\in\WW^{1,2}(\Om,T),$ consider the direct sum
$$ Y=Y_0+Y_1,\quad Y_0\in\N(\B),\quad Y_1\in D(\B).$$ Then there is $Z\in D(\B)\cap\R(\B)$ such that $Y_1=\div DZ.$ We have
\beq W(Y)&&=W(Y_1)=\int_\Om\<W,\div DZ\>dg=-\int_\Om\<DW,DZ\>dg+\int_\Ga DZ(W,\nu)d\Ga\nonumber\\
&&\leq\|DW\|_{\WW^{-1,2}(\Om,T)}\|DZ\|_{\WW^{1,2}(\Om,T)}+\|DZ\|_{\LL^2(\Ga,T)}\|W\|_{\LL^2(\Ga,T)}\nonumber\\
&&\leq C(\|DW\|_{\WW^{-1,2}(\Om,T)}+\|W\|_{\LL^2(\Ga,T)})\|Z\|_{\WW^{2,2}(\Om,T)}\nonumber\\
&&\leq C(\|DW\|_{\WW^{-1,2}(\Om,T)}+\|W\|_{\LL^2(\Ga,T)})\|Y_1\|_{\LL^2(\Om,T)}\nonumber\\
&&\leq C(\|DW\|_{\WW^{-1,2}(\Om,T)}+\|W\|_{\LL^2(\Ga,T)})\|Y\|_{\LL^2(\Om,T)}\qflq Y\in\LL^2(\Om,T),\nonumber\eeq that is, (\ref{3.30}).

{\bf Step 2.}\,\,\,Let $\Phi_1,$ $\cdots,$ $\Phi_k$ be an orthonormal basis of $\N(\B),$ where $k=\dim\N(\B).$ Now we claim there are constants $\si>0$ and $C>0$ such that
\beq&&\si[(\sum_{i=1}^k\a_k^2)^{1/2}+\|DW\|_{\WW^{-1,2}(\Om,T^2)}]\nonumber\\
&&\leq\|\sum_{i=1}^k\a_kD\Phi_i+DW\|_{\WW^{-1,2}(\Om,T^2)}+\|W\|_{\LL^2(\Ga,T)}\nonumber\\
&&\leq C[(\sum_{i=1}^k\a_k^2)^{1/2}+\|DW\|_{\WW^{-1,2}(\Om,T^2)}+\|W\|_{\LL^2(\Ga,T)}]\label{3.31}\eeq for all $(\a_1,\cdots,\a_k)\in\R^k$ and $W\in\R(\B).$ Clearly, the right hand side of (\ref{3.31}) holds true. We prove the left hand side by contradiction. Suppose the left hand side of (\ref{3.31}) is not true. Then there exist $\{\,(\a_{j1},\cdots,\a_{jk})\,\}\subset\R^k$ and $W_j\in\R(\B)$ such that
$$(\sum_{i=1}^k\a_{ji}^2)^{1/2}+\|DW_j\|_{\WW^{-1,2}(\Om,T^2)}=1,$$
$$\|\sum_{i=1}^k\a_{ji}D\Phi_i+DW_j\|_{\WW^{-1,2}(\Om,T^2)}+\|W_j\|_{\LL^2(\Ga,T)}\leq\frac1j\qflq j\geq1.$$
We assume that $(\a_{j1},\cdots,\a_{jk})\rw(\a_{1},\cdots,\a_{k}).$ Then
$$W_j\rw 0\qiq\LL^2(\Ga,T),\quad DW_j\rw U\qiq \WW^{-1,2}(\Om,T^2).$$ By Step 1
$$W_j\rw W_0\qiq\LL^2(\Om,T)$$ with $W_0|_\Ga=0.$ Thus
$$(\sum_{i=1}^k\a_{i}^2)^{1/2}+\|DW_0\|_{\WW^{-1,2}(\Om,T^2)}=1,\quad D(\sum_i\a_i\Phi_i+W_0)=0,\quad (\sum_i\a_i\Phi_i+W_0)|_\Ga=0.$$
Then $\sum_i\a_i\Phi_i+W_0=0$ and, so
$$(\a_{1},\cdots,\a_{k})=0,\quad W_0=0.$$ That is a contradiction.

{\bf Step 3.}\,\,\,Let $W\in\LL^2(\Om,T)$ be given. Let
$$W=\sum_i\a_i\Phi_i+W_0,\quad \a_i=(W,\Phi_i)_{\LL^2(\Om,T)},\quad W_0\in\R(\B).$$
For given $Y\in\WW^{1,2}(\Om,T),$ we have by Steps 1 and 2
\beq W(Y)&&=\sum_i\a_i\b_i+W_0(Y)\leq|\a||\b|+C(\|DW_0\|_{\WW^{-1,2}(\Om,T^2)}+\|W_0\|_{\LL^2(\Ga,T)})\|Y\|_{\LL^2(\Om,T)}\nonumber\\
&&\leq C(|\a|+\|DW_0\|_{\WW^{-1,2}(\Om,T^2)}+\|W_0\|_{\LL^2(\Ga,T)})\|Y\|_{\LL^2(\Om,T)}\nonumber\\
&&\leq C(\|DW\|_{\WW^{-1,2}(\Om,T^2)}+\|W\|_{\LL^2(\Ga,T)})\|Y\|_{\LL^2(\Om,T)},\label{3.32}\eeq where
$$\a=(\a_1,\cdots,\a_k),\quad\b=(\b_1,\cdots,\b_k),\quad \b_i=(Y,\Phi_i)_{\LL^2(\Om,T)}.$$

Finally, (\ref{3.29}) follows from (\ref{3.32}).
\end{proof}

\begin{lem}\label{3l.2} The following estimate holds.
\be \|v\|_{[\WW^{1,2}(\Om)]'}\leq C(\|Dv\|_{[\WW^{2,2}(\Om,T)]'}+\|v\|_{[\WW^{3/2,2}(\Ga)]'})\qfq v\in[\WW^{1,2}(\Om)]'.\label{3x.5}\ee
\end{lem}

\begin{proof} For given $f\in\WW^{1,2}(\Om),$ we solve problem
$$-\Delta z=f\qfq x\in\Om,\quad z|_\Ga=0.$$ Then
$$\|z\|_{\WW^{3,2}(\Om)}\leq C\|f\|_{\WW^{1,2}(\Om)}.$$
Thus we have
\beq(v,f)_{\LL^2(\Om)}&&=-\int_\Ga v\<Dz,\nu\>d\Ga+(Dv,Dz)_{\LL^2(\Om,T)}\nonumber\\
&&\leq\|v\|_{[\WW^{3/2,2}(\Ga)]'}\|Dz\|_{\WW^{3/2,2}(\Ga,T)}+\|Dv\|_{[\WW^{2,2}(\Om,T)]'}\|Dz\|_{\WW^{2,2}(\Ga,T)}\nonumber\\
&&\leq C(\|Dv\|_{[\WW^{2,2}(\Om,T)]'}+\|v\|_{[\WW^{3/2,2}(\Ga)]'}))\|f\|_{\WW^{1,2}(\Om)}.\nonumber\eeq
Thus (\ref{3x.5}) follows.
\end{proof}

\begin{lem}\label{3l.3} Let $\Om\subset M$ be a region and  boundary $\Ga$ consist of finite many closed curves. Then
\be\|\div\div P\|_{[\WW^{2,2}(\Om)]'}\leq C\|P\|_{\LL^2(\Om,T^2)}\qfq P\in\LL^2(\Om,T^2).\label{3n.21}\ee
\end{lem}

\begin{proof}
Assume $P\in \WW^{2,2}(\Om,T^2).$  From \cite[Lemma 2.3]{CY}, we have
\be \div(PZ)=\<P,DZ\>+\<\div P, Z\>\qfq x\in\Om,\quad Z\in\WW^{1,2}(\Om,T).\label{3x.22}\ee
For given $z\in\WW^{2,2}(\Om),$ it follows from (\ref{3x.22}) that
\beq(\div\div P, z)_{\LL^2(\Om)}&&=\int_\Om[\div(z\div P)-\<\div P,Dz\>]dg\nonumber\\
&&=\int_\Ga (z\<\div P,\nu\>-\<PDz,\nu\>)d\Ga+\int_\Om\<P,D^2z\>dg\nonumber\\
&&\leq C(\|\div P\|_{[\WW^{3/2,2}(\Ga,T)]'}\|z\|_{\WW^{3/2,2}(\Ga)}\nonumber\\
&&\quad+\|P\|_{[\WW^{1/2,2}(\Ga,T^2)]'}\|Dz\|_{\WW^{1/2,2}(\Ga,T)}+\|P\|_{\LL^2(\Om,T^2)}\|D^2z\|_{\LL^2(\Om,T^2)}\nonumber\\
&&\leq C\|P\|_{\LL^2(\Om,T^2)}\|z\|_{\WW^{2,2}(\Om)}.\label{3n.22}\eeq
In addition, the assumption that boundary $\Ga$ consists of finite many closed curves implies
\be[\WW^{\lam,2}(\Ga,T^k)]'=\WW^{-\lam,2}(\Ga,T^k),\label{3x.24}\ee where $\lam>0$ is a real number and $k\geq0$ is an integer.

Since $\WW^{2,2}(\Om,T^2)$ is dense in $[\WW^{2,2}(\Om,T^2)]',$ (\ref{3n.21}) follows from (\ref{3n.22}) and (\ref{3x.24}).
\end{proof}

Suppose that for given $\varepsilon>0$ small, $\eta_0\in\CC^\infty_0(-\infty,\infty)$  is given satisfying
$$0\leq\eta_0\leq1;\quad \eta_0=0\qfq s\leq \varepsilon/2;\quad \eta_0=1\qfq s\geq\varepsilon.$$ Set
$$\eta(\a(t,s))=s\eta_0(s)\qfq (t,s)\in(0,a)\times(-b_0,b_1).$$
For given $\varepsilon>0$ small, let
$$\Om_{-\varepsilon}=\{\,\a(t,s)\,|\,(t,s)\in(0,a]\times(-\varepsilon,b_1)\,\}.$$
Suppose
$$X\in\CC^m(\Om_{-\varepsilon},T)$$ is given in \cite[Section 2, (2.34)]{CY1}.  We define
\be\L_0V
=e^{-\g\kappa}[(\div Q\nabla\n V-\eta\<V,QX\>)QX+(\div V-\eta\<V,\nabla\n X\>)\nabla\n X],\ee
\beq\L_0^*V&&=\nabla\n Q D(e^{-\g\kappa}\<QX,V\>)-D(e^{-\g\kappa}\<\nabla\n X,V\>)\nonumber\\
&&\quad-\eta e^{-\g\kappa}(\<V,QX\>QX+\<V,\nabla\n X\>\nabla\n X),\eeq for $V\in\LL^2(\Om_{-\varepsilon},T).$ Let $m\geq0$ be an integer. By \cite[Lemma 2.3]{CY1} we have
\beq (W,\L_0V)_{\LL^2(\Om,T)}&&=(V,\L_0^*W)_{\LL^2(\Om,T)}\nonumber\\
&&\quad+\int_{\pl\Om}(\pounds_1\<V,\nu\>\<W,\nu\>-\pounds_2\<V,QX\>\<W,QX\>)d\Ga\label{3x.7}\eeq
for $V,$ $W\in T\Om,$ where
\be\pounds_1=\frac{e^{-\g\kappa}}{\<X,\nu\>}\Pi(X,X),\quad \pounds_2=\frac{e^{-\g\kappa}}{\<X,\nu\>}\Pi(Q\nu,Q\nu) \qfq x\in\pl\Om.\label{n2.17x}\ee

Consider problems
\be\left\{\begin{array}{l}\L_0V=F,\\
\<V,QX\>|_{\Ga_{-\varepsilon}}=p,\quad\<V,\nu\>|_{\Ga_{b_1}}=q\end{array}\right.\label{3x.9}\ee and
\be\left\{\begin{array}{l}\L_0^*V=F,\\
\<V,\nu\>|_{\Ga_{-\varepsilon}}=p,\quad\<V,QX\>|_{\Ga_{b_1}}=q,\end{array}\right.\label{3x.10}\ee respectively, where $F\in\WW^{m,2}(\Om_{-\varepsilon},T),$ $p\in\WW^{m,2}(\Ga_{-\varepsilon}),$ and
$q\in\WW^{m,2}(\Ga_{b_1})$ are given.
It follows from \cite[Theorems 3.1 and 3.2]{CY1} that
\begin{thm}\label{3t.1} Let $S$ be of $\CC^{m+3}.$ There exists a unique solution $W\in\WW^{m,2}(\Om_{-\varepsilon},T)$ to problem $(\ref{3x.9}),$ or $(\ref{3x.10})$ satisfying
\beq &&\|V\|_{\WW^{m,2}(\Om_{-\varepsilon},T)}+\|V\|_{\WW^{m,2}(\pl\Om_{-\varepsilon}, T)}\nonumber\\
&&\leq C( \|F\|_{\WW^{m,2}(\Om_{-\varepsilon},T)}+ \|q\|_{\WW^{m,2}(\Ga_{b_1})}+ \|p\|_{\WW^{m,2}(\Ga_{-\varepsilon})}).\label{V3.2}\eeq
\end{thm}

We define linear operators $\L_0:$ $\WW^{m,2}(\Om_{-\varepsilon},T)\rw\WW^{m,2}(\Om_{-\varepsilon},T)$ by
\be\L_0V
=e^{-\g\kappa}[(\div Q\nabla\n V)QX+(\div V)\nabla\n X]-\C V,\label{x3.1}\ee
\beq D(\L_0)=\{\,V\in \WW^{m,2}(S,T)\big|&&\,\div V\in \WW^{m,2}(\Om_{-\varepsilon}), \div Q\nabla\n V\in \WW^{m,2}(\Om_{-\varepsilon})\nonumber\\
&&\,\<V,QX\>|_{\Ga_{-\varepsilon}}=\<V,\nu\>|_{\Ga_{b_1}}=0\,\},\nonumber\eeq
where
\be\C V=\eta e^{-\g\kappa}(\<V,QX\>QX+\<V,\nabla\n X\>\nabla\n X),\quad \eta(x)=s\eta_0(s)\label{n3.1}\ee  for $x=\a(t,s)\in\Om_{-\varepsilon}.$
Moreover, set
\be\L_0^*V=\nabla\n Q D(e^{-\g\kappa}\<QX,V\>)-D(e^{-\g\kappa}\<\nabla\n X,V\>)-\C V.\label{n3.2}\ee
$$ D(\L_0^*)=\{\,V\in \WW^{m,2}(S,T)\big|\,\L_0^*V\in\WW^{m,2}(\Om_{-\varepsilon},T),\,\<V,\nu\>|_{\Ga_{-\varepsilon}}=\<V,QX\>|_{\Ga_{b_1}}=0\,\}.$$

By Theorem \ref{3t.1} and \cite[Theorems 3.3 and 3.4]{CY1} we have the following.
\begin{thm}\label{3t.2}
$(i)$\,\,\,Operators $\L_0$ and $\L_0^*$ have bounded inverses on $\WW^{m,2}(\Om_{-\varepsilon},T)$ with the estimates
\beq&&\|\L_0^{-1}V\|_{\WW^{m,2}(\Om_{-\varepsilon},T)}+\|\L_0^{-1}V\|_{\WW^{m,2}(\pl\Om_{-\varepsilon}, T)}\quad\mbox{or}\nonumber\\
&&\|{\L_0^*}^{-1}V\|_{\WW^{m,2}(\Om_{-\varepsilon},T)}+\|{\L_0^*}^{-1}V\|_{\WW^{m,2}(\pl\Om_{-\varepsilon}, T)}\leq C\|V\|_{\WW^{m,2}(\Om_{-\varepsilon},T)}\eeq for $V\in\WW^{m,2}(\Om_{-\varepsilon},T).$

$(ii)$\,\,\,Operator $\L^{-1}_0\C$ and ${\L_0^*}^{-1}\C:$ $\WW^{m,2}(\Om_{-\varepsilon},T)\rw\WW^{m+1,2}(\Om_{-\varepsilon},T)$ are compact.
\end{thm}

Define operators $\L$ and $\L^*:$ $\D'_m(S,T)\rightarrow \D'_{m+1}(S,T)$ by
$$\L V=\L_0 V+\C V\quad\mbox{and}\quad \L^*V=\L_0^*V+\C V,$$ respectively.
Given $Z_0\in \E_m(\Om_{-\varepsilon},T)$, consider problem
\be
\begin{cases}
\L^*Z+\nabla\n Dz=Z_0\qfq x\in\Om_{-\varepsilon},\\
-\Delta z=e^{-\g\kappa}\Pi(X,Z)\tr\Pi\qfq x\in\Om_{-\varepsilon},\\
\pl_\nu z\big|_{\pl\Om_{-\varepsilon}}=
\<Z,\nu\>\big|_{\Ga_{-\varepsilon}}=\<Z,QX\>\big|_{\Ga_{b_1}}=0.
\end{cases}\label{1.33}
\ee

We define operator $\A:$ $\LL^2(\Om,T)\rw\WW^{1,2}(\Om,T)$ by
$$\A Z=\nabla\n Dz\qfq Z\in\LL^2(\Om,T),$$ where $z\in \WW^{2,2}(\Om,T)$ is given by
$$\left\{\begin{array}{l}-\Delta z=e^{-\g\kappa}\<Z,\nabla\n X\>\tr\Pi\qfq x\in\Om,\\
\pl_\nu z=0\quad\mbox{on}\quad x\in\Ga.\end{array}\right.$$ Then
$$\A:\quad \LL^2(\Om,T)\rw\LL^2(\Om_{-\varepsilon},T)$$ is compact. Thus by Theorem \ref{3t.2}
$${\L_0^*}^{-1}(\C+\A):\quad \LL^2(\Om,T)\rw\LL^2(\Om_{-\varepsilon},T)$$ is also compact. Set
$$\N=\{\,W\in\LL^2(\Om_{-\varepsilon},T)\,|\,W+{\L_0^*}^{-1}(\C+\A)W=0\,\},$$
$$\N_*=\{\,W\in\LL^2(\Om_{-\varepsilon},T)\,|\,W+\L_0^{-1}(\C+\A^*)W=0\,\}.$$ Then
$$\dim \N<\infty,\quad \dim \N_*<\infty.$$

\begin{pro}\label{p3.1}  Let $S$ be of $\CC^{m+3}.$
Problem $(\ref{1.33})$ admits a unique solution $Z\perp \N$ in $\LL^2(\Om_{-\varepsilon},T)$ if and only if  $Z_0\perp \N_*$ in $\LL^2(\Om_{-\varepsilon},T).$ Moreover, there exists $C>0$ such that if
$Z_0\in\WW^{2,2}(\Om_{-\varepsilon},T)$ and $(Z,z)$ solves problem $(\ref{1.33})$ with $Z\perp \N,$  then $(Z,z)\in\WW^{m,2}(\Om_{-\varepsilon},T)\times\WW^{m+2,2}(\Om_{-\varepsilon})$ satisfying
\be\|Z\|^2_{\WW^{m,2}(\Om_{-\varepsilon},T)}+\|z\|^2_{\WW^{m+2,2}(\Om_{-\varepsilon})}+\|Z\|^2_{\WW^{m,2}(\pl\Om_{-\varepsilon},T)}\leq C\|Z_0\|^2_{\WW^{m,2}(\Om_{-\varepsilon},T)}.\label{3.2}\ee
\end{pro}

\begin{proof} For given $Z_0\in\LL^2(\Om_{-\varepsilon},T),$ it is easy to check that $Z\in\LL^2(\Om_{-\varepsilon},T)$ solves problem (\ref{1.33}) if and only if $Z$ satisfies
$$Z+{\L_0^*}^{-1}(\C+\A)Z={\L_0^*}^{-1}Z_0\qiq\LL^2(\Om_{-\varepsilon},T).$$
By Fredholm's theorem the above problem has a solution $Z\in\LL^2(\Om_{-\varepsilon},T)$ if and only if $Z_0\perp \N_*$ in $\LL^2(\Om_{-\varepsilon},T).$

Let $Z_0\in\WW^{m,2}(\Om_{-\varepsilon},T).$ By \cite[Theorems 3.1 and 3.2]{CY1}  $Z\in\WW^{m,2}(\Om_{-\varepsilon},T).$ Thus by the regularity of the elliptic problem $z\in\WW^{m+2,2}(\Om_{-\varepsilon})$ and the estimates
$$\|z\|^2_{\WW^{m+2,2}(\Om_{-\varepsilon})}\leq C\|Z\|^2_{\WW^{m,2}(\Om_{-\varepsilon},T)}$$
follows. Moreover, using \cite[Lemma 3.2 and Theorem 2.2]{CY1}, we have
$$\|Z\|^2_{\WW^{m,2}(\Om_{-\varepsilon},T)}+\|Z\|^2_{\WW^{m,2}(\pl\Om_{-\varepsilon},T)}\leq C\|Z_0\|^2_{\WW^{m,2}(\Om_{-\varepsilon},T)}.$$ Thus (\ref{3.2}) follows.
\end{proof}

Let $\Phi_1,$ $\cdots,$ $\Phi_k$ be an orthonormal basis of $\N,$ where $k=\dim\N.$ Define the projection operator ${\P}:$ $[\WW^{m,2}(\Om_{-\varepsilon},T)]'\rw\N$ by
$${\P}V=\sum_{j=1}^k(V,\Phi_j)_{\LL^2(\Om_{-\varepsilon})}\Phi_k\qfq V\in[\WW^{m,2}(\Om_{-\varepsilon},T)]'.$$

 Consider problem
 \be \left\{\begin{array}{l}Dv=\nabla\n V+F\qfq x\in \Om_{-\varepsilon},\\
\div V=-(\tr\Pi)v+f\qfq x\in \Om_{-\varepsilon},\\
\<V,QX\>|_{\Ga_{-\varepsilon}}=p,\quad\<V,\nu\>|_{\Ga_{b_1}}=q.\end{array}\right.\label{3.11n}\ee

We have the following.

\begin{thm}\label{t3.3} Let $S$ be of $\CC^5.$ Suppose that $(v,V)$ solves problem $(\ref{3.11n}).$ Then there exists a $C>0$ such that
\beq&&\|V\|_{[\WW^{2,2}(\Om_{-\varepsilon},T)]'}+\|v\|_{[\WW^{1,2}(\Om_{-\varepsilon})]'}\leq C\Big(\|{\P}V\|_{[\WW^{2,2}(\Om_{-\varepsilon},T)]'}+I(F,f,p,q)\Big),\quad\label{3x.25}\eeq where
\beq I(F,f,p,q)&&=\|\div QF\|_{[\WW^{2,2}(\Om_{-\varepsilon})]'}+\|F\|_{[\WW^{2,2}(\Om_{-\varepsilon},T)]'}+\|f\|_{[\WW^{2,2}(\Om_{-\varepsilon})]'}\nonumber\\
&&\quad
+\|p\|_{\WW^{-2,2}(\Ga_{-\varepsilon})}+\|q\|_{\WW^{-2,2}(\Ga_{b_1})}.\nonumber\eeq
\end{thm}

\begin{proof} From (\ref{3.11n}) we have
\be \left\{\begin{array}{l}\div Q\nabla \n V=-\div QF\qfq x\in \Om_{-\varepsilon},\\
\div V=-\rho v+f\qfq x\in \Om_{-\varepsilon},\end{array}\right.\label{3x.26}\ee where $\rho=\tr\Pi.$
It follows from (\ref{x3.1}), (\ref{n3.1}),  and (\ref{3x.26}) that
\be\L_0V+\C V=G,\label{3x.27}\ee where
$$G=-e^{-\g\kappa}[(\div QF)QX+(\rho v-f)\nabla\n X].$$

Let
\be\Xi=\{\, Z\in\WW^{2,2}(\Om_{-\varepsilon},T)\,|\,Z\perp\N_*\,\mbox{in $\LL^2(\Om_{-\varepsilon},T)$\,}\}.\label{3n.27}\ee For given $Z_0\in\Xi,$ we solve problem (\ref{1.33}) to have the solution $(Z, z)\in\WW^{2,2}(\Om_{-\varepsilon},T)\times\WW^{4,2}(\Om_{-\varepsilon}).$ Let
 $$z_1=e^{-\g\kappa}\<Z,QX\>,\quad z_2=\rho e^{-\g\kappa}\<Z,\nabla\n X\>,\quad z_3=e^{-\g\kappa}\<Z,\nabla\n X\>.$$ By (\ref{1.33}), (\ref{3x.7}), and the first equation in
 (\ref{3.11n}), we  have
 \beq(V,Z_0)_{\LL^2(\Om_{-\varepsilon},T)}&&=(V,\L_0^*Z+\C Z+\nabla\n Dz)_{\LL^2(\Om_{-\varepsilon},T)}=(G, Z)_{\LL^2(\Om_{-\varepsilon},T)}\nonumber\\
 &&\quad+(Dv-F,Dz)_{\LL^2(\Om_{-\varepsilon},T)}+\int_{\Ga_{-\varepsilon}}\pounds_2p\<Z,QX\>d\Ga\nonumber\\
 &&\quad-\int_{\Ga_{b_1}}\pounds_1q\<Z,\nu\>d\Ga\nonumber\\
 &&=-(\div QF,z_1)_{\LL^2(\Om_{-\varepsilon})}+(f, z_3)_{\LL^2(\Om_{-\varepsilon})}\nonumber\\
 &&\quad-(F,Dz)_{\LL^2(\Om_{-\varepsilon},T)}-(v,z_2)_{\LL^2(\Om_{-\varepsilon})}+(Dv,Dz)_{\LL^2(\Om_{-\varepsilon})}\nonumber\\
 &&\quad+\int_{\Ga_{-\varepsilon}}\pounds_2p\<Z,QX\>d\Ga-\int_{\Ga_{b_1}}\pounds_1q\<Z,\nu\>d\Ga.\nonumber\eeq
Noting that
$$-(v,z_2)_{\LL^2(\Om_{-\varepsilon})}+(Dv,Dz)_{\LL^2(\Om_{-\varepsilon})}=0,$$
by (\ref{3.2}) with $m=2$ and Lemma \ref{3l.3}  we obtain
\beq&&|(V,Z_0)_{\LL^2(\Om_{-\varepsilon},T)}|\leq\|\div QF\|_{[\WW^{2,2}(\Om_{-\varepsilon})]'}\|z_1\|_{\WW^{2,2}(\Om_{-\varepsilon})}\nonumber\\
&&\quad+\|f\|_{[\WW^{2,2}(\Om_{-\varepsilon})]'}\|z_2\|_{\WW^{2,2}(\Om_{-\varepsilon})}+\|F\|_{[\WW^{2,2}(\Om_{-\varepsilon},T)]'}\|Dz\|_{\WW^{2,2}(\Om_{-\varepsilon},T)}\nonumber\\
&&\quad+C(\|p\|_{\WW^{-2,2}(\Ga_{-\varepsilon})}+\|q\|_{\WW^{-2,2}(\Ga_{b_1})})\|Z\|_{\WW^{2,2}(\pl\Om_{-\varepsilon},T)}\nonumber\\
&&\leq CI(F,f,p,q)\|Z_0\|_{\WW^{2,2}(\Om_{-\varepsilon},T)}\qflq Z_0\in\Xi.\label{n3x.28}\eeq

Since $\Xi$ is a closed subspace of $\WW^{2,2}(\Om_{-\varepsilon},T),$ (\ref{n3x.28}) shows that there exists a function in $\Xi,$ denoted by ${\B}V,$ such that
$$(V,Z_0)_{\LL^2(\Om_{-\varepsilon},T)}=({\B}V,Z_0)_{\WW^{2,2}(\Om_{-\varepsilon},T)}\qflq Z_0\in\Xi.$$

Let $\II$ be the canonical map from $[\WW^{2,2}(\Om_{-\varepsilon},T)]'\rw\WW^{2,2}(\Om_{-\varepsilon},T).$ Let $V_0={\II}^{-1}{\B}V.$ Then $V_0\in[\WW^{2,2}(\Om_{-\varepsilon},T)]'$ satisfies
$$(V-V_0,Z_0)_{\LL^2(\Om_{-\varepsilon},T)}=0\qflq Z_0\in\Xi.$$ Thus
\be V-V_0\in\N_*.\label{3x.28}\ee Therefore,
$$V={\P} V+V_0\in[\WW^{2,2}(\Om_{-\varepsilon},T)]',$$ where ${\P}V=V-V_0.$ By (\ref{n3x.28}) we obtain
\beq\|V\|_{[\WW^{2,2}(\Om_{-\varepsilon},T)]'}&&\leq\|{\P} V\|_{[\WW^{2,2}(\Om_{-\varepsilon},T)]'}+\|V_0\|_{[\WW^{2,2}(\Om_{-\varepsilon},T)]'}\nonumber\\
&&=\|{\P} V\|_{[\WW^{2,2}(\Om_{-\varepsilon},T)]'}+\sup_{Z_0\in\WW^{2,2}(\Om_{-\varepsilon},T), \|Z_0\|_{\WW^{2,2}(\Om_{-\varepsilon},T)}=1}}(V,Z_0)_{\LL^2(\Om_{-\varepsilon},T)\nonumber\\
&&\leq \|{\P} V\|_{[\WW^{2,2}(\Om_{-\varepsilon},T)]'}+CI(F,f,p,q).\label{3x.29}\eeq

For given $z_0\in\WW^{1,2}(\Om_{-\varepsilon}),$ we solve problem
$$\Delta w=z_0\qfq x\in\Om_{-\varepsilon},\quad \<Dw,\nu\>|_{\pl\Om_{-\varepsilon}}=0.$$ Then
$$\|w\|_{\WW^{3,2}(\Om_{-\varepsilon})}\leq C\|z_0\|_{\WW^{1,2}(\Om_{-\varepsilon})}\qflq z_0\in\WW^{1,2}(\Om_{-\varepsilon}).$$
By the first equation in (\ref{3.11n}) we have
\beq(v,z_0)_{\LL^2(\Om_{-\varepsilon})}&&=-(Dv,Dw)_{\LL^2(\Om_{-\varepsilon},T)}=-(\nabla\n V+F,Dw)_{\LL^2(\Om_{-\varepsilon},T)}\nonumber\\
&&\leq C(\|V\|_{[\WW^{2,2}(\Om_{-\varepsilon},T)]'}+\|F\|_{[\WW^{2,2}(\Om_{-\varepsilon},T)]'})\|z_0\|_{\WW^{1,2}(\Om_{-\varepsilon})}\nonumber\eeq for all $z_0\in\WW^{1,2}(\Om_{-\varepsilon}),$ that yields
\be \|v\|_{[\WW^{1,2}(\Om_{-\varepsilon})]'}\leq C(\|V\|_{[\WW^{2,2}(\Om_{-\varepsilon},T)]'}+\|F\|_{[\WW^{2,2}(\Om_{-\varepsilon},T)]'}).\label{n3x.31}\ee

Thus (\ref{3x.25}) follows from (\ref{3x.29}) and (\ref{n3x.31}).
\end{proof}

\begin{thm}\label{t3.4} Let
$$U=\sym DW+w\Pi\qfq y=W+w\n.$$ Then there exists $C>0,$ independent of $y=W+w\n,$ such that
\be\|W\|_{\LL^2(\Om_{-\varepsilon},T)}+\|w\|_{[\WW^{1,2}(\pl\Om_{-\varepsilon})]'}\leq C(\|U\|_{\LL^2(\Om_{-\varepsilon},T^2)}+\|W\|_{\LL^2(\pl\Om_{-\varepsilon},T)}).\label{3.26}\ee
\end{thm}

\begin{proof}Let $V$ be given by (\ref{V2.29}). By Theorem \ref{t2.1} $(W,w)$ and $(v,V)$ satisfy problems
\be v=\frac12\div QW,\label{n3xx.33}\ee
\be\left\{\begin{array}{l}Dw=\nabla\n W-QV\qfq x\in \Om_{-\varepsilon},\\
\div W=-\rho w+\tr U\qfq x\in \Om_{-\varepsilon},\end{array}\right.\label{3n.34}\ee and
\be\left\{\begin{array}{l}Dv=\nabla\n V+Q\div QUQ\qfq x\in \Om_{-\varepsilon},\\
\div V=-\rho v-\<Q\nabla\n,U\>\qfq x\in \Om_{-\varepsilon},\end{array}\right.\label{n3x.34}\ee respectively.
Then
$$\div Q\nabla\n W=-\div V=\rho v-\<Q\nabla\n,U\>\qfq x\in\Om_{-\varepsilon}.$$

Applying Theorem \ref{t3.3} to problem (\ref{3n.34}) with
$$F=-QV,\quad f=\tr U,\quad p=\<W,QX\>|_{\Ga_{-\varepsilon}},\quad q=\<W,\nu\>|_{\Ga_{b_1}},$$ we have
\beq&&\|W\|_{[\WW^{2,2}(\Om_{-\varepsilon},T)]'}+\|w\|_{[\WW^{1,2}(\Om_{-\varepsilon})]'}\nonumber\\
&&\leq C\Big(\|{\P}W\|_{[\WW^{2,2}(\Om_{-\varepsilon},T)]'}+\|\div V\|_{[\WW^{2,2}(\Om_{-\varepsilon})]'}+\|V\|_{[\WW^{2,2}(\Om_{-\varepsilon},T)]'}+\|\tr U\|_{[\WW^{2,2}(\Om_{-\varepsilon})]'}\nonumber\\
&&\quad
+\|\<W,QX\>\|_{\WW^{-2,2}(\Ga_{-\varepsilon})}+\|\<W,\nu\>\|_{\WW^{-2,2}(\Ga_{b_1})}\Big).\label{3x.35}\eeq
Then applying Theorem \ref{t3.3} to problem (\ref{n3x.34}) with
$$F=Q\div QUQ,\quad f=-\<Q\nabla\n,U\>,\quad p=\<V,QX\>,\quad q=\<V,\nu\>$$ yields
\beq&&\|\div V\|_{[\WW^{2,2}(\Om_{-\varepsilon})]'}+\|V\|_{[\WW^{2,2}(\Om_{-\varepsilon},T)]'}\nonumber\\
&&\leq \|\<Q\nabla\n,U\>\|_{[\WW^{2,2}(\Om_{-\varepsilon})]'}+C(\|V\|_{[\WW^{2,2}(\Om_{-\varepsilon},T)]'}+\|v\|_{[\WW^{1,2}(\Om_{-\varepsilon})]'})\nonumber\\
&&\leq C(\|{\P}V\|_{[\WW^{2,2}(\Om_{-\varepsilon},T)]'}+\|U\|_{\LL^2(\Om_{-\varepsilon},T)}\nonumber\\
&&\quad+\|\<V,QX\>\|_{\WW^{-2,2}(\Ga_{-\varepsilon})}+\|\<V,\nu\>\|_{\WW^{-2,2}(\Ga_{b_1})}),\label{3x.36}\eeq where Lemma \ref{3l.3} has been used. In addition, it follows from the first equation in (\ref{3n.34}) that
\be{\P}V={\P}Q(Dw-\nabla\n W).\label{3x.37}\ee

Using (\ref{3x.36}) and (\ref{3x.37}) in (\ref{3x.35}) and by Lemma \ref{l3.4} below, we have
\beq&&\|W\|_{[\WW^{2,2}(\Om_{-\varepsilon},T)]'}+\|w\|_{[\WW^{1,2}(\Om_{-\varepsilon})]'}\leq C\Big(\|{\P}W\|_{[\WW^{2,2}(\Om_{-\varepsilon},T)]'}\nonumber\\
&&\quad+\|U\|_{\LL^2(\Om_{-\varepsilon},T^2)}+\|W\|_{\LL^2(\pl\Om_{-\varepsilon},T)}
+\|{\P}Q(Dw-\nabla\n W)\|_{\WW^{-2,2}(\pl\Om_{-\varepsilon},T)}\Big).\quad\label{n3x.37}\eeq
We claim that the terms $\|{\P}W\|_{[\WW^{2,2}(\Om_{-\varepsilon},T)]'}$ and $\|{\P}Q(Dw-\nabla\n W)\|_{[\WW^{2,2}(\Om_{-\varepsilon},T)]'}$ can be removed in (\ref{n3x.37}) to have
\beq&&\|W\|_{[\WW^{2,2}(\Om_{-\varepsilon},T)]'}+\|w\|_{[\WW^{1,2}(\Om_{-\varepsilon})]'}\leq C\Big(\|U\|_{\LL^2(\Om_{-\varepsilon},T^2)}+\|W\|_{\LL^2(\pl\Om_{-\varepsilon},T)}\Big)\quad\label{n3x.38}\eeq by a compactness-uniqueness argument as follows.

By contradiction. Suppose that (\ref{n3x.38}) does not hold true. Then there are $w_k,$ $W_k,$ $v_k,$  $V_k,$ and $U_k,$ which satisfy (\ref{n3xx.33})-(\ref{n3x.34}), respectively, such that
\be1=\|W_k\|_{\LL^2(\Om_{-\varepsilon},T)}+\|w_k\|_{[\WW^{1,2}(\pl\Om_{-\varepsilon})]'}\geq k(\|U_k\|_{\LL^2(\Om_{-\varepsilon},T^2)}+\|W_k\|_{\LL^2(\pl\Om_{-\varepsilon},T)})\ee for all $k\geq1.$
Then
$$U_k\rw0\qiq\LL^2(\Om_{-\varepsilon},T^2),\quad W_k\rw0\qiq\LL^2(\pl\Om_{-\varepsilon},T).$$ By Lemma \ref{l3.4} below
$$V_k\rw0\qiq\WW^{-2,2}(\pl\Om_{-\varepsilon}),\quad w_k\rw0\qiq\WW^{-1,2}(\pl\Om_{-\varepsilon}).$$ Noting that $\div V_k=-\rho v_k-\<Q\nabla\n, U_k\>,$ by (\ref{3x.36}) there is a subsequence, denoted stll by $(v_k,V_k),$ that converges to some
$(v_0,V_0)\in[\WW^{2,2}(\Om_{-\varepsilon})]'\times[\WW^{2,2}(\Om_{-\varepsilon},T)]'$ in $[\WW^{2,2}(\Om_{-\varepsilon})]'\times[\WW^{2,2}(\Om_{-\varepsilon},T)]'.$
It follows from (\ref{n3x.34}) that
\be\left\{\begin{array}{l}Dv_0=\nabla\n V_0\qfq x\in \Om_{-\varepsilon},\\
\div V_0=-\rho v_0\qfq x\in \Om_{-\varepsilon},\\
V_0|_{\pl\Om_{-\varepsilon}}=0.\end{array}\right.\label{n3x.40}\ee
By (\ref{n3x.37}) there exists a subsequence $(w_k,W_k)$ which satisfies
$$(w_k,W_k)\rw (w_0,W_0)\qiq [\WW^{1,2}(\Om_{-\varepsilon})]'\times[\WW^{2,2}(\Om_{-\varepsilon},T)]'$$ for some $(w_0,W_0)\in[\WW^{1,2}(\Om_{-\varepsilon})]'\times[\WW^{2,2}(\Om_{-\varepsilon},T)]'.$
It is easy to check that $(w_0,W_0)$ solves problem
\be \left\{\begin{array}{l}\|W_0\|_{\LL^2(\Om_{-\varepsilon},T)}+\|w_0\|_{[\WW^{1,2}(\Om_{-\varepsilon})]'}=1,\\
Dw_0=\nabla\n W_0-QV_0\qfq x\in \Om_{-\varepsilon},\\
\div W_0=-\rho w_0\qfq x\in \Om_{-\varepsilon},\\
w_0|_{\pl\Om_{-\varepsilon}}=0,\quad W_0|_{\pl\Ga_{-\varepsilon}}=0,\end{array}\right.\label{n3x.41}\ee
$$\sym DW_0+w_0\Pi=0\qfq x\in\Om_{-\varepsilon},$$ and
$$v_0=\frac12\div QW_0\qfq x\in\Om_{-\varepsilon},$$ respectively.

Consider problem (\ref{n3x.40}). Let $\tau=\a_t/|\a_t|$ for $x\in\Ga_{b_1}.$ Suppose  $\nu=Q\tau.$  Then $\nu,$ $\tau$ has positive orientation. We have
\beq2v_0&&=\div QW_0=\<D_\nu W_0,\tau\>-\<D_\tau W_0,\nu\>=\<D_\nu W_0,\tau\>+\<D_\tau W_0,\nu\>\nonumber\\
&&=-2w_0\Pi(\tau,\nu)=0\qfq x\in\Ga_{b_1}.\nonumber\eeq
It follows from (\ref{n3x.40}) that
\be\left\{\begin{array}{l}\div(\nabla\n)^{-1}Dv_0=-\rho v_0\qfq x\in \Om_{-\varepsilon},\\
v_0=Dv_0=0\qfq x\in\Ga_{b_1}.\end{array}\right.\ee
By the uniqueness of the elliptic problem
$$v_0=0\qfq x\in\Om_{-\varepsilon},\quad\kappa(x)>0.$$
Moreover,  a similar argument yields
$$v_0=0\qfq x\in\Ga_{-\varepsilon}.$$ Then we have a hyperbolic problem
$$\left\{\begin{array}{l}\div(\nabla\n)^{-1}Dv_0=-\rho v_0\qfq x\in \Om_{-\varepsilon},\quad\kappa(x)<0,\\
v_0=Dv_0=0\qfq x\in\Ga_{-\varepsilon}.\end{array}\right.$$
Since the curves
$$\a(\cdot,s)\qfq -\varepsilon\leq s\leq 0,$$ are not characteristic, the uniqueness of the hyperbolic problem implies
$$v_0=0\qfq x\in\Om_{-\varepsilon},\quad \kappa(x)<0.$$
Thus we obtain
$$v_0=V_0=0\qfq x\in\Om_{-\varepsilon}.$$ Moreover, by (\ref{n3x.41}) we have
$$\left\{\begin{array}{l}\div(\nabla\n)^{-1}Dw_0=-\rho w_0\qfq x\in \Om_{-\varepsilon},\\
w_0=Dw_0=0\qfq x\in\pl\Om_{-\varepsilon}.\end{array}\right.$$
By a similar argument as for $(v_0,V_0)$ we obtain
$$w_0=W_0=0\qfq x\in\Om_{-\varepsilon},$$ which contradicts the first equality in (\ref{n3x.41}).

Finally, noting that
$$\|DW\|_{[\WW^{1,2}(\Om_{-\varepsilon},T^2)]'}=\|\sym DW\|_{[\WW^{1,2}(\Om_{-\varepsilon},T^2)]'},$$ by  Lemma \ref{l3.1} and (\ref{n3x.38}) we have
\beq\|W\|_{\LL^2(\Om_{-\varepsilon},T)}&&\leq C(\|\sym DW\|_{[\WW^{1,2}(\Om_{-\varepsilon},T^2)]'}+\|W\|_{\LL^2(\pl\Om_{-\varepsilon},T)})\nonumber\\
&&\leq C(\|U\|_{[\WW^{1,2}(\Om_{-\varepsilon},T^2)]'}+\|w\|_{[\WW^{1,2}(\Om_{-\varepsilon})]'}+\|W\|_{\LL^2(\pl\Om_{-\varepsilon},T)})\nonumber\\
&&\leq \Big(\|U\|_{\LL^2(\Om_{-\varepsilon},T^2)}+\|W\|_{\LL^2(\pl\Om_{-\varepsilon},T)}\Big).\label{3xx.45}\eeq
Thus (\ref{3.26}) follows from (\ref{n3x.38}) and (\ref{3xx.45}).
\end{proof}

\begin{lem} \label{l3.4} Suppose that $(w,W)$ and $(v,V)$ solve problem $(\ref{3n.34})$ and $(\ref{n3x.34}),$ respectively. Then
\be \|V\|_{\WW^{-2,2}(\pl\Om_{-\varepsilon})}+\|w\|_{\WW^{-1,2}(\pl\Om_{-\varepsilon})}\leq C(\|U\|_{\LL^2(\Om_{-\varepsilon},T^2)}+\|W\|_{\LL^2(\pl\Om_{-\varepsilon},T)}
).\label{3x.37}\ee
\end{lem}

\begin{proof} Let $\tau=\a_t/|\a_t|$ for $x\in\Ga_{b_1},$ Then $\tau,$ $\nu$ forms an orthonormal basis along $\Ga_{b_1}.$ We have
\be w\Pi(\tau,\tau)=U(\tau,\tau)-DW(\tau,\tau)\qfq x\in\Ga_{b_1}.\label{3x.40}\ee Since $\Pi(\tau,\tau)\not=0$ for all $x\in\Ga_{B_1},$ it follows from (\ref{3x.40}) that
\beq\|w\|_{\WW^{-1,2}(\pl\Om_{-\varepsilon})}&&\leq C(\|U\|_{\WW^{-1,2}(\pl\Om_{-\varepsilon},T)}+\|\tau\<W,\tau\>\|_{\WW^{-1,2}(\pl\Om_{-\varepsilon})}+\|\<W,D_\tau\tau\>\|_{\WW^{-1,2}(\pl\Om_{-\varepsilon})})\nonumber\\
&&\leq C(\|U\|_{\LL^2(\Om_{-\varepsilon},T^2)}+\|W\|_{\LL^2(\pl\Om_{-\varepsilon},T)}).\label{3x.41}\eeq

Moreover, by (\ref{3x.40}) we obtain
\beq \nu(w)\Pi(\tau,\tau)&&=\nu(U(\tau,\tau))-D^2W(\tau,\tau,\nu)-\<D_\tau\tau,\nu\>[DW(\nu,\tau)+DW(\tau,\nu)]-w\nu(\Pi(\tau,\tau))\nonumber\\
&&=\nu(U(\tau,\tau))-D^2W(\tau,\tau,\nu)-2\<D_\tau\tau,\nu\>U(\tau,\nu)\nonumber\\
&&\quad+w[2\Pi(\tau,\nu)-\nu(\Pi(\tau,\tau))]\qfq x\in\Ga_{b_1}.\label{3x.38}\eeq
Then
\beq\|\nu(w)\|_{\WW^{-2,2}(\Ga_{b_1})}&&\leq C(\|U\|_{\LL^2(\Om_{-\varepsilon},T)}+\|w\|_{\WW^{-1,2}(\Ga_{b_1})}\nonumber\\
&&\quad+\|D^2W(\tau,\tau,\nu)\|_{\WW^{-2,2}(\Ga_{b_1})}).\label{3x.39}\eeq
Next, we shall estimate $\|D^2W(\tau,\tau,\nu)\|_{\WW^{-2,2}(\Ga_{b_1})}.$ By Ricci's identity, we have
\beq D^2W(\tau,\tau,\nu)&&=D^2W(\tau,\nu,\tau)+R(\tau,\nu,W,\tau)=\tau(DW(\tau,\nu))\nonumber\\
&&\quad-\<D_\tau\tau,\nu\>DW(\nu,\nu)-\<D_\tau\nu,\tau)DW(\tau,\tau)+R(\tau,\nu,W,\tau)\nonumber\\
&&=-\tau(DW(\nu,\tau))+2\tau(U(\tau,\nu))-2\tau(w\Pi(\tau,\nu))+R(\tau,\nu,W,\tau)\nonumber\\
&&\quad-\<D_\tau\tau,\nu\>[U(\nu,\nu)-w\Pi(\nu,\nu)]-\<D_\tau\nu,\tau)[U(\tau,\tau)-w\Pi(\tau,\tau)],\nonumber\eeq which yields
$$ \|D^2W(\tau,\tau,\nu)\|_{\WW^{-2,2}(\Ga_{b_1})}\leq C(\|U\|_{\LL^2(\Om_{-\varepsilon},T)}+\|W\|_{\LL^2(\Ga_{b_1})}+\|w\|_{\WW^{-1,2}(\Ga_{b_1})}),$$ where $R(\cdot,\cdot,\cdot,\cdot)$ is the Riemmannian curvature tensor.
It follows from (\ref{3x.39}) that
$$\|\nu(w)\|_{\WW^{-2,2}(\Ga_{b_1})}\leq C(\|U\|_{\LL^2(\Om_{-\varepsilon},T)}+\|W\|_{\LL^2(\Ga_{b_1})}+\|w\|_{\WW^{-1,2}(\Ga_{b_1})}).$$
By the first equation in (\ref{3n.34}) we obtain
\beq\|V\|_{\WW^{-2,2}(\Ga_{b_1},T)}&&\leq C(\|W\|_{\WW^{-2,2}(\Ga_{b_1},T)}+\|\nu(w)\|_{\WW^{-2,2}(\Ga_{b_1})}+\|\tau(w)\|_{\WW^{-2,2}(\Ga_{b_1})})\nonumber\\
&&\leq C(\|U\|_{\LL^2(\Om_{-\varepsilon},T)}+\|W\|_{\LL^2(\Ga_{b_1})}+\|w\|_{\WW^{-1,2}(\Ga_{b_1})}).\label{3x.44}\eeq

Combing (\ref{3x.41}) and (\ref{3x.44}), we have
$$\|V\|_{\WW^{-2,2}(\Ga_{b_1},T)}+\|w\|_{\WW^{-1,2}(\Ga_{-\varepsilon})}\leq C(\|U\|_{\LL^2(\Om_{-\varepsilon},T)}+\|W\|_{\LL^2(\Ga_{-\varepsilon})}).$$
A similar argument shows that the above estimates hold when $\Ga_{b_1}$ is replaced by $\Ga_{-\varepsilon}.$
Thus (\ref{3x.37}) follows.
\end{proof}

{\bf Proof of Theorem \ref{t1.1}.}\,\,\,Let $y=W+w\n$ satisfy problem (\ref{s}) for given $U\in\LL^2(S,T^2).$ Let
$$S({-b_0},-\varepsilon)=\{\,\a(t,s)\in S\,|\,(t,s)\in(0,a)\times(-b_0,-\varepsilon)\,\}.$$ By (\ref{Pi2.801}) $S({-b_0},-\varepsilon)$ is a  non-characteristic region. By \cite[Theorem 1.1]{Yao2019} there is a $C>0,$ independent of $y,$ such that
\be\|W\|_{\LL^2(S({-b_0},-\varepsilon),T)}\leq C(\|U\|_{\LL^2(S({-b_0},-\varepsilon),T^2)}+\|W\|_{\LL^2(\Ga_{-b_0},T)}),\label{3.52}\ee from that we obtain
\beq\|w\|_{[\WW^{1,2}(S({-b_0},-\varepsilon))]'}&&\leq C(\|U\|_{[\WW^{1,2}(S({-b_0},-\varepsilon),T^2)]'}+\|\sym DW\|_{[\WW^{1,2}(S({-b_0},-\varepsilon),T^2)]'})\nonumber\\
&&\leq  C(\|U\|_{\LL^2(S({-b_0},-\varepsilon),T^2)}+\|W\|_{\LL^2(S({-b_0},-\varepsilon),T)})\nonumber\\
&&\leq C(\|U\|_{\LL^2(S({-b_0},-\varepsilon),T^2)}+\|W\|_{\LL^2(\Ga_{-b_0},T)}),\label{3.53}\eeq where the assumption
$$\Pi(\a_t,\a_t)>0\qflq x\in\overline{S({-b_0},-\varepsilon)}$$ is used. Moreover, by \cite[Lemma 3.6]{Yao2019}
\be\|W\|_{\LL^2(\Ga_{-\varepsilon},T)}\leq C(\|U\|_{\LL^2(S({-b_0},-\varepsilon),T^2)}+\|W\|_{\LL^2(\Ga_{-b_0},T)}).\label{3.54}\ee
Thus estimate (\ref{1.3}) follows from Theorem \ref{t3.4} and (\ref{3.52})-(\ref{3.54}). \hfill$\Box$\\

{\bf Ethical approval}

This article does not contain any studies with human participants or animals performed by the author.

{\bf Declaration of competing interest}

The author declares that there is no conflict of interest.

 \end{document}